\documentclass[11pt]{article}
\usepackage{hyperref}
\usepackage{amssymb,epsfig,graphics, graphicx, url,times,mathrsfs,algorithm,algorithmic,amsmath,array,latexsym,xspace,fancyhdr,wrapfig, xcolor,multirow,amsthm,caption,cite,helvet,sidecap}

\usepackage[normalem]{ulem}

\usepackage[inline]{enumitem}
\usepackage{arydshln}
\usepackage[makeroom]{cancel}
\usepackage{varwidth}

\usepackage{tikz}
\usetikzlibrary{shapes.geometric}
\usetikzlibrary{shapes,snakes,patterns}
\usetikzlibrary{arrows,positioning,automata,calc}
\usetikzlibrary{plotmarks}

\definecolor{green1}{rgb}{0,0.5,0}
\definecolor{magenta}{rgb}{1.0, 0.11, 0.81}
\definecolor{mulberry}{rgb}{0.77, 0.29, 0.55}
\definecolor{xgray}{rgb}{0.9, 0.9, 0.9}
\def \blue{\color{blue}}
\def \red{\color{red}}

\def \bes{\begin{equation*}}
\def \ees{\end{equation*}}
\def \bas{\begin{align*}}
\def \eas{\end{align*}}
\def \be{\begin{equation}}
\def \ee{\end{equation}}
\def \bbm{\begin{bmatrix}}
\def \ebm{\end{bmatrix}}
\def \cf{\mathcal{S}}
\def \ck{\mathcal{R}}
\def \cR{\mathcal{R}_2}
\def \cd{\mathcal{D}}
\def \ce{\mathcal{W}}

\def \mbs{\mathbf{s}}
\def \CO{\texttt{CO}}
\def \RO{\texttt{RO}}
\def \SC{\texttt{SC}}
\def \mrS{\mathrm{S}}
\def \mrW{\mathrm{W}}
\def \mrR {\mathrm{R}}

\newtheorem{theorem}{Theorem}
\newtheorem{corollary}{Corollary}

\newtheorem{example}{Example}
\newtheorem{remark}{Remark}

\begin{document}

\title{Staircase Codes for Secret Sharing  with Optimal Communication and Read  Overheads}
\author{Rawad Bitar and Salim El Rouayheb\\ ECE Department, IIT, Chicago\\ Emails: rbitar@hawk.iit.edu, salim@iit.edu}

\renewcommand{\today}{}
\maketitle
\begin{abstract}
We study the communication efficient secret sharing (CESS) problem. A classical threshold  secret sharing scheme encodes a secret  into $n$ shares given to $n$ parties, such that any set of at least $t$, $t<n$,  parties can reconstruct the secret, and any set of at most $z$, $z<t$, colluding parties cannot obtain any information about the secret. A CESS scheme satisfies the previous properties of threshold secret sharing. Moreover, it allows  to reconstruct the secret from any set of $d,\ d\geq t,$ parties by reading and communicating the minimum  amount of information. In this paper, we introduce two explicit constructions of CESS codes called  {\em Staircase Codes}. The first construction achieves optimal communication and read costs for a fixed $d,\ d\geq t$. The second construction achieves optimal costs universally for all possible values of $d, t\leq d\leq n$. Both constructions are designed over a small finite field $GF(q)$, for any prime power $q> n$. We also describe how Staircase codes can also be used to construct threshold changeable secret sharing with minimum storage cost, i.e., minimum share size.
\end{abstract}


\section{Introduction}
Consider the  threshold  secret sharing (SS) problem \cite{S79,Blakley1979} in which a dealer encodes a secret using random keys into $n$ shares and distributes them to $n$ parties. The threshold SS allows a legitimate user contacting any set of at least $t$, $t<n$,  parties to reconstruct the secret by downloading their shares. In addition, the scheme ensures that any set of at most $z$, $z<t$, colluding parties cannot obtain any information, in an information theoretic sense, about the secret. The following example illustrates the construction of a threshold SS on $n=4$ shares.

\begin{example}[Threshold SS]
\label{ex:SS1}
Let $n=4, t=2$ and $z=1$ and let $s$ be a secret uniformly distributed over $GF(5)$. Then, the following $4$ shares $(s+r, s+2r,s+3r,s+4r)$ form a threshold SS scheme, with $r$ being a  random symbol, called key, chosen uniformly at random from $GF(5)$ and independently of $s$. A user can decode the secret by contacting any $t=2$ parties, downloading their shares and decoding $s$ and $r$. Secrecy is ensured, because the secret is padded with the key in each share.
\end{example}

Threshold secret sharing code constructions have been extensively studied in the literature, e.g., \cite{S79,Blakley1979,McESa81,CJT00,KGH83,LPD04,CYH04}. The  literature on secret sharing predominantly studies non-threshold secret sharing schemes, with  so-called general access structures, e.g., \cite{BL88,ISN89,B89}. We refer the interested reader to the following survey works \cite{P12,be11,CDN15} and references within. In this paper, we focus  on the problem of communication (and read) efficient secret sharing (CESS). A CESS scheme satisfies the  properties of  threshold secret sharing described in the previous paragraph. In addition, it achieves minimum communication and read overheads when the user contacts $d,\ d\geq t,$ parties. The communication overhead ({\CO}) is defined as the extra amount of information (beyond the secret size) downloaded by a user contacting $d$ parties in order to decode the secret. The read overhead {\RO} is defined similarly. Next, we give an example of a CESS code that minimizes {\CO} and {\RO}. The CESS code in this example belongs to the family of Staircase codes which we introduce in Section~\ref{sec:cons1}.

\begin{example}
\label{ex:intro} 
Consider again the SS problem of Example~\ref{ex:SS1} with $n=4$, $t=2$, $z=1$. We assume now that the secret $\mathbf{s}$ is formed  of 2 symbols $s_1,s_2$ uniformly distributed over $GF(5)$ and we use two keys  $r_1,r_2$  drawn independently and uniformly at random from $GF(5)$.  To construct the Staircase code,  the secret symbols and keys are arranged in a matrix $M$ as shown in~\eqref{eq:encex}. The matrix $M$ is multiplied by a $4\times3$ Vandermonde matrix $V$ to obtain the matrix $C=VM$. The $4$ rows of $C$ form the $4$ different shares and give the Staircase\footnote{The nomenclature of Staircase codes  comes from the position of the zero block matrices in the general structure of the matrix $M$ (see the general construction in  Table~\ref{tab:unim}).} code shown in Table~\ref{tab:intro}. \be
\label{eq:encex}
\begin{tikzpicture}[baseline=(current  bounding  box.center)]
\tikzstyle{stealth} = [draw=none,text=black]

\draw[draw=none,fill=xgray] (0,-0.69) -- (0,-0.22) -- (0.63,-0.22) -- (0.63,-0.69) -- cycle;
\node[stealth] (1) at (0,0)
{$\bbm s_1 &  r_1\\
s_2 & r_2\\ 
r_1 & 0\\ \ebm$};

\node[stealth] (3) [right=-0.2cm of 1] {.};
\node[stealth] (2) [left=-0.1cm of 1]  {$\bbm 1&1&1\\
1&2&4\\
1&3&4\\
1&4&1\\ \ebm$};

\draw[thin, gray,decorate,decoration={brace,amplitude=10pt,mirror}]  (-2.6,-0.96) -- (-1.08,-0.96);
\node[stealth,color=gray] at (-1.85,-1.5) {$V$};

\draw[thin, gray,decorate,decoration={brace,amplitude=7pt,mirror}]  (-0.64,-0.725) -- (0.64,-0.725);
\node[stealth,color=gray] at (0,-1.15) {$M$};

\draw[thin] (-0.63,-0.69) -- (0,-0.69) -- (0,-0.23) -- (0.63,-0.23);
\draw[thin] (0,-.69) -- (0.63,-0.69);

\node[stealth] [left=1.8cm of 1] {$C=VM=$};

%
 \end{tikzpicture}
\ee

\begin{table}[h!]
\normalsize
\centering
\begin{tabular}[h!]{c|c|c|c}
Party 1& Party 2 & Party 3 & Party 4\\ \hline
\blue $s_1+s_2+ r_1$ & \blue $s_1+2s_2+4r_1$& \blue $s_1+3s_2+ 4r_1$& \blue $s_1+4s_2+r_1$ \\
$r_1+r_2$ & $r_1+2r_2$ & $ r_1+3r_2$ & $ r_1+4r_2$ \\
\end{tabular}
\vspace{0.1cm}
\caption{An example of a CESS code based on the Staircase code construction over $GF(5)$ for $n=4$ parties, threshold $t=2$, $z=1$ colluding parties and any $d=3$ parties can efficiently reconstruct the secret. A user contacting any $t=2$ parties downloads all their shares, i.e., $4$ symbols in total, in order to decode the secret. The resulting overheads are $\CO=\RO=2$ symbols. However, a user contacting any $d=3$ parties decodes the secret by downloading the first symbol (in blue) of each share, i.e., $3$ symbols in total. Hence, $\CO=\RO=1$ symbol. For instance, a user contacting parties $1,2$ and $3$ downloads $s_1+s_2+ r_1$,  $s_1+2s_2+4r_1$, and $s_1+3s_2+ 4r_1$ and can decode the secret and $r_1$. Notice that a user contacting $d=3$ parties can only decode $r_1$, whereas a user contacting $t=2$ parties has to decode $r_1$ and $r_2$.}
\label{tab:intro}

\end{table}

The CESS scheme enjoys the following properties. First, a user decodes the secret either by contacting any $t=2$ parties and downloading all their shares, i.e., 4 symbols, or by contacting any $d=3$ parties and downloading the first symbol (in blue) of each share, i.e., $3$ symbols in total. The key idea here is that the user is only interested in decoding the secret and not necessarily the keys. When $d=3$, the user decodes the secret and only the key $r_1$, whereas when $d=t=2$, the user has to decode the secret and both of the keys.  This code actually achieves the minimum {\CO} and {\RO} equal to $1$ symbol for $d=3$ (and 2 symbols for $d=t=2$) given later in~\eqref{eq:CO} and~\eqref{eq:RC}. Second, secrecy is achieved because the secret $s_1,s_2$ is padded by random keys $r_1,r_2$ and each $z=1$ party cannot obtain any information about $s_1$ and $s_2$. \end{example}

\noindent{\em Related work:} The CESS problem was introduced by Wang and Wong in \cite{WW08} where they focused on perfect CESS, i.e., the case in which $z=t-1$. The authors showed that there exists a tradeoff between the number of contacted parties $d$ and the amount of information downloaded by a user in order to decode the secret. They derived a lower bound on {\CO} and constructed codes for the special case of $z=t-1$ using polynomial evaluation over $GF(q)$, where $q>n+v$, that achieve minimum {\CO} and {\RO} universally for all $d,\ t\leq d\leq t+v-1$, for some positive integer $v$. Zhang et al.\cite{ZYSMH12} constructed CESS codes for the special case of $z=t-1$ over $GF(q)$, where $q>n$, that achieve minimum {\CO} and {\RO} for any fixed $d$, $t\leq d \leq n$. Recently, Huang et al.\cite{HLKB15} studied the CESS problem for all $z<t$ and generalized the lower bound on {\CO}. The authors constructed explicit CESS codes for any $z$ achieving the minimum {\CO} and {\RO} for $d=n$ over $GF(q),\ q>n(n-z)$. Moreover, they proved the achievability of the lower bound on {\CO} and {\RO} universally for all possible values of $d, t\leq d \leq n$ using random linear code constructions\footnote{After the appearance of the original version of this work on Arxiv \cite{BR15}, an equivalent CESS code construction for all parameters was given independently in \cite{HLKB16}.}. In our setting, we assume that the dealer has direct access to all the parties. In the case where the dealer can access the parties through a network, Shah et al. \cite{SRR15} studied the problem of minimizing the communication cost of securely delivering the shares to the parties.
  
\noindent{\em Contributions:} 
In this paper, we introduce two new classes of explicit constructions of linear CESS codes that achieve minimum {\CO} and {\RO}. More specifically, we make the following contributions:

\begin{enumerate}
\item We describe a construction, which we call {\em Staircase Code}, that achieves minimum {\CO} and {\RO}  for any given $z$ and any given $d$. This construction generalizes the construction in Example~\ref{ex:intro}.
\item We describe a universal construction, which we call {\em Universal Staircase Code}, that achieves minimum {\CO} and {\RO} simultaneously for all possible values of $d,\ t\leq d\leq n$ and any given value of $z$.
\end{enumerate}
Moreover, we describe how to construct a class of secret sharing codes, called threshold changeable secret sharing (TCSS) codes \cite{MPSH99}, based on the introduced Staircase codes.

\noindent The Staircase codes require a small finite field $GF(q)$ of size $q> n$, which is the same requirement for Reed Solomon based SS codes\footnote{However, the constructions requires to divide  the secret into a certain number of  symbols, which may not be necessary for SS codes.} \cite{McESa81}.

\vspace{0.2cm}

\noindent{\em Organization:} 
The paper is organized as follows. In section~\ref{sec:sys}, we formulate the CESS problem, introduce the necessary notations and summarize our results. We describe the Staircase code constructions in section~\ref{sec:cons}. In section~\ref{sec:fds}, we prove that the Staircase codes for a fixed $d$ achieve secrecy and minimum {\CO} and {\RO}. In section~\ref{sec:uns}, we prove that Universal Staircase codes achieve secrecy and minimum {\CO} and {\RO}. In Section~\ref{sec:tcss} we show how to use the Staircase codes to construct threshold changeable secret sharing. We conclude in section~\ref{sec:conc}.

\section{Problem formulation and main results}\label{sec:sys}
We consider the CESS problem and   follow the majority of the notations  in \cite{HLKB15}. A secret $\mathbf{s}$ of size $k$ units is formed of $k\alpha$ symbols (1 unit = $\alpha$ symbols). The secret symbols are  drawn independently and uniformly at random from a finite alphabet, typically a finite field. A CESS code is a scheme that encodes the secret, using random keys, into $n$ shares $w_1,\dots, w_n$, of unit size each,  and distributes them to $n$ distinct parties. Let $\mrW_i$ denote the random variable representing the share of party $i$, let $\mrS$ denote the random variable representing the secret $\mbs$, let $[n]=\{1,\dots,n\}$, and for any subset $B\subseteq[n]$ denote by $\mrW_B$ the set of random variables representing the shares indexed by $B$, i.e., $\mrW_B=\{\mrW_i; i\in B\}$. Then, a CESS code must satisfy the following properties:
\begin{enumerate}
\item {\em Perfect secrecy:} Any subset of $z$ or less parties  should not be able to get any information about the secret. The perfect secrecy condition can be expressed as 
\begin{equation}\label{eq:secrecy}
H(\mrS\mid \mrW_Z)=H(\mrS),\, \forall Z\subset [n] \text{ s.t.}   \left\lvert Z \right\rvert = z.
\end{equation}

\item {\em MDS:} A user downloading any $t$ shares is able to recover the secret, i.e., 
\be 
\label{eq:mds}
H(\mrS \mid \mrW_{A})=0,\, \forall A\subseteq [n] \text{ s.t.}   \left\lvert A\right\rvert=t.
\ee
Equations~\eqref{eq:secrecy} and \eqref{eq:mds}  imply that the secret can be of at most $t-z$ units (see \cite[Proposition~1]{HLKB15}). We will take the secret to be of maximum size, i.e.,   $k=t-z$ units.
\item {\em Minimum {\CO} and {\RO}:} a user contacting any $d$ parties, $t\leq d\leq n$, is able to decode the secret by reading and downloading exactly $k+\CO(d)$ units of information in total from all the contacted shares, where 
\be 
\label{eq:CO}
\CO(d)=\dfrac{kz}{d-z}.\ee

\end{enumerate} 


\noindent Equation~\eqref{eq:CO}  represents the achievable information theoretic lower bound \cite[Theorem~3.1]{WW08},~\cite[Theorem~1]{HLKB15} on the   communication overhead, $\CO(d)$, needed to satisfy the constraints in \eqref{eq:secrecy} and \eqref{eq:mds}, when the user contacts $d$ parties\footnote{Note that a user contacting $d$ parties and achieving \eqref{eq:CO} for a threshold secret sharing with threshold $t$ downloads the same amount of information as a user contacting $d$ parties in a threshold secret sharing with threshold $d$.}. Since the amount of information read cannot be less than the downloaded amount, the following lower bound on {\RO} holds,
\begin{equation}\label{eq:RC}
\RO(d)\geq \CO(d).
\end{equation}
We will refer to a CESS code described above as an $(n,k,z,d)$ CESS code, where the threshold is $t=k+z$. For instance, the code in Example~\ref{ex:intro} is an $(4,1,1,3)$ CESS code.  We define a universal $(n,k,z)$ CESS code that  achieves minimum $\CO(d)$ and $\RO(d)$ simultaneously for all possible values of $d$.
Note that the MDS constraint can be omitted since it is subsumed by the minimum {\CO} and {\RO} constraint since it corresponds to the case of $d=t$ and  $\CO(t)=z$. However, we will make this distinction for clarity of exposition.

\vspace{0.2cm}

\noindent Given the model described above, we are ready to state our two main results. 


\begin{theorem}\label{thm:main}
The $(n,k,z,d)$ Staircase CESS code defined in Section~\ref{sec:cons1} over $GF(q)$, $q>n$, satisfies the required MDS and perfect secrecy constraints given in~\eqref{eq:secrecy}~and~\eqref{eq:mds}, and achieves optimal communication and read overheads $\CO(d)$ and $\RO(d)$ given in~\eqref{eq:CO}~and~\eqref{eq:RC} for any given $d,\ d\in\{k+z,\dots,n\}$.
\end{theorem}

\begin{theorem}\label{thm:main2}
The $(n,k,z)$ Universal Staircase CESS code defined in Section~\ref{sec:cons2} over $GF(q)$, $q>n$, satisfies the required MDS and perfect secrecy constraints given in~\eqref{eq:secrecy}~and~\eqref{eq:mds}, and achieves optimal communication and read overheads $\CO(d)$ and $\RO(d)$ given in~\eqref{eq:CO}~and~\eqref{eq:RC} simultaneously for all $d,$ $k+z\leq d\leq n$.
\end{theorem}

\section{Staircase code constructions} \label{sec:cons}

\subsection{Staircase code construction for fixed $d$} \label{sec:cons1}
We describe the $(n,k,z,d)$ Staircase code construction that achieves optimal communication and read overheads $\CO(d)$ and $\RO(d)$ for any given $d$, $k+z\leq d \leq n$. In this construction, we take $\alpha=d-z$.  Hence, the secret $\mathbf{s}$  of size $k$ units is formed of $k(d-z)$ symbols $s_1,\dots,s_{k\alpha},$ where $s_i\in GF(q)$ and $ q> n$. The symbols $s_i$ are arranged in an $\alpha \times k$ matrix $\cf$. The construction uses $z\alpha$ iid random keys drawn uniformly at random from $GF(q)$ and independently of the secret. The keys are partitioned into two matrices $\mathcal{R}_1$ and $\mathcal{R}_2$ of dimensions $z\times k$ and $z\times(\alpha-k)$ respectively. Let $\cal{D}$ be the transpose of the last $(\alpha-k)$ rows of the matrix $\begin{bmatrix} \mathcal{S}\\ \mathcal{R}_1\end{bmatrix}$\footnote{If $\alpha-k\leq z$, i.e., $d\leq 2z+k$, then $\mathcal{D}$ consists of the transpose of the last $\alpha-k$ rows of $\mathcal{R}_1$.} and let $\mathbf{0}$ be the all zero square matrix of dimensions $(\alpha-k)\times(\alpha-k)$, note that $\alpha-k\geq 0$ since $d\geq z+k$. The key ingredient of the construction is to arrange the secret and the keys in a $d\times \alpha$ matrix $M$ defined in Table~\ref{tab:M}. The inspiration behind this construction is the class of Product Matrix codes that minimizes the repair bandwidth in distributed storage systems\footnote{After the appearance of the original version of this work on Arxiv \cite{BR15}, a connection between the family of regenerating codes and CESS codes was explored in more details in \cite{RKV16}.} \cite{RSK11}.
\begin{table}[h!]\label{tab:M}
\normalsize
\linespread{1}
\centering
  \resizebox{0.3\textwidth}{!}{
  \begin{tikzpicture}[baseline=(current  bounding  box.center)]
\tikzstyle{stealth} = [draw=none,text=black]

\draw[draw=none,fill=xgray] (0.463,-0.91) -- (0.463,-0.48) -- (1.38,-0.48) -- (1.38,-.91) -- cycle;
\node[stealth] (1) at (0,0){$
M=\left[\!\begin{array}[h!]{c:c}
\multirow{2}{*}{$\cf$}&\mathcal{D}\\ \cdashline{2-2}
& \multirow{2}{*}{$\cR$}\\ \cdashline{1-1}
\multirow{2}{*}{$\mathcal{R}_1$} & ~\\ \cdashline{2-2}
& \mathbf{0}\\
\end{array}\!\right]$};
\node[stealth,gray] (4) at (-1.35,-0.35) {$\scriptstyle d\times \alpha$};
\draw[<->,gray,gray] (0.45,1.1) to (1.35,1.1);
\node[stealth,gray] (2) at (0.9,1.3) {$\scriptstyle \alpha-k$};
\draw[<->,gray,gray] (-0.45,1.1) to (0.45,1.1);
\node[stealth,gray] (2) at (0.05,1.3) {$\scriptstyle k$};
\draw[<->,gray] (-0.55,0) to (-0.55,0.95);
\node[stealth,gray] (3) at (-0.7,0.45) {$\scriptstyle \alpha$};
\draw[<->,gray] (-0.55,-0.93) to (-0.55,0);
\node[stealth,gray] (3) at (-0.7,-0.5) {$\scriptstyle z$};
\draw[<->,gray] (1.45,0.47) to (1.45,0.95);
\node[stealth,gray] (3) at (1.6,0.7) {$\scriptstyle k$};
\draw[<->,gray] (1.45,-0.47) to (1.45,0.47);
\node[stealth,gray] (3) at (1.6,0) {$\scriptstyle z$};
\filldraw[black] (1.9,0) circle (0.5pt); 
\draw[<->,gray] (1.45,-0.93) to (1.45,-0.47);
\node[stealth,gray] (3) at (1.8,-0.65) {$\scriptstyle \alpha-k$};
\draw[thin] (-0.42,-0.91) -- (0.464,-0.91) -- (0.464,-0.477) -- (1.36,-0.477);
\draw[thin] (0.465,-0.91) -- (1.36,-0.91);
\draw[color=white,fill=white] (0.44,-0.927) -- (0.47,-0.927) -- (0.47,-0.99) -- cycle;

%
%
%

  \end{tikzpicture}
  }
  \caption{The structure of the matrix $M$ that contains the secret  and keys in the Staircase code construction for fixed $d$.}
  \label{tab:M}
\end{table}

\vspace{0.2cm}

\noindent{\em Encoding: }Let $V$ be an $n\times d$ Vandermonde\footnote{We require all square sub-matrices formed by consecutive columns of $V$ to be invertible. Vandermonde and Cauchy matrices satisfy this property.} matrix defined over $GF(q)$. The matrix $M$, defined in Table~\ref{tab:M}, is multiplied by $V$ to obtain the matrix $C=VM$. The $n$ rows of $C$ form the $n$ different shares.
\vspace{0.2cm}

\noindent{\em Decoding: } A user contacting any $t=k+z$ parties downloads all the  shares  of the contacted parties.  A user contacting  $d$ parties, indexed by $I\subseteq[n]$, downloads the first $k$ symbols from each contacted party corresponding to $v_i \bbm \cf & \ck_1\ebm^t, i\in I$ (the superscript $t$ denotes the transpose of  a matrix). Theorem~\ref{thm:main} guaranties that the user will be able to decode the secret in both cases.

\addtocounter{example}{-1}
\begin{example}[Continued]
 \label{sec:fde}
We give the  details of  the construction of the $(n,k,z,d)=(4,1,1,3)$ CESS code of Example~\ref{ex:intro}. We take $\alpha=d-z=2$, thus the secret $\mathbf{s}$ is formed of $k\alpha=2$ symbols $s_1,\ s_2$ uniformly distributed over $GF(q)$, $q=5>n=4$. The construction uses $z\alpha=2$ iid random keys $r_1,\ r_2$ drawn uniformly at random over $GF(5)$ and independently of the secret. The keys are partitioned into two matrices $\ck_1$ and $\ck_2$ of dimensions $z\times k=1\times 1$ and $z\times (\alpha-k)=1 \times 1$, respectively. The matrix $\cal{D}$ is the transpose of the last $\alpha-k=1$ row of $\mathcal{R}_1$. Hence, we have, $\mathcal{R}_1= \mathcal{D}=  r_1,\  \mathcal{R}_2= r_2,$ and $\mathcal{S}=\bbm
s_1\\
s_2 \\
\ebm.$
The secret and the keys are arranged in a $d\times \alpha=3\times 2$ matrix $M$. Let $V$ be an $n\times d=4\times 3$ Vandermonde matrix. $M$ and $V$ are given again in~\eqref{eq:fdm}. \be
\label{eq:fdm}
M=\begin{bmatrix}
s_1 &  r_1\\
s_2 & r_2\\ 
r_1 & 0\\
\end{bmatrix} \text { and }
V=\begin{bmatrix}
1&1&1\\
1&2&4\\
1&3&4\\
1&4&1\\
\end{bmatrix}.
\ee
%
%


\noindent The shares are the rows of the matrix $C=VM$ as given in Table~\ref{tab:intro}.
We want to check that this code  satisfies the following properties: 

\vspace{0.2cm}

\noindent{\em 1) Minimum {\CO} and {\RO} for $d=3$:} 
We check that a user contacting $d=3$ parties can reconstruct the secret with minimum {\CO} and {\RO}. For instance, if a user contacts the first 3 parties and downloads the first symbol of each contacted share, then the downloaded data is given by,
\be
\label{eq:d3}
\begin{bmatrix}
1&1&1\\
1&2&4\\
1&3&4\\
\ebm
\begin{bmatrix}
s_1\\
s_2 \\ 
r_1\\
\end{bmatrix} .
\ee
The matrix on the left is a $3\times 3 $ square Vandermonde matrix, hence invertible. Therefore, the user can  decode the secret (and $r_1$). This remains true irrespective of which $3$ parties are contacted.
The user reads and downloads $3$ symbols of size $3/\alpha=3/2$ units resulting in minimum overheads, $\CO(3)=\RO(3)=3/2-k=1/2$, as  given in~\eqref{eq:CO}~and~\eqref{eq:RC}.

\vspace{0.2cm}

\noindent{\em 2) MDS:} We check that a user contacting $t=k+z=2$ parties can reconstruct the secret. Suppose the user contacts parties $1$ and $2$ and downloads all their shares given by  
\be
\label{eq:d2}
\bbm
1&1&1\\
1&2&4\\
\ebm \bbm
s_1 & r_1\\
s_2 & r_2\\
r_1 & 0\\
\ebm.
\ee
The system in~\eqref{eq:d2} is equivalent to the two following systems $
\bbm
1&1&1\\
1&2&4\\
\ebm\begin{bmatrix}
s_1\\
s_2 \\ 
r_1\\
\end{bmatrix} $ and 
 $\bbm
1&1\\
1&2\\
\ebm \bbm r_1\\r_2\ebm.$
The decoder uses the latter system  to decode $r_1$ and $r_2$. This is possible because the matrix on the left is a  square Vandermonde matrix, hence invertible.
Then, the decoder subtracts the obtained value of $r_1$ from the former system to obtain again the following invertible system $ \bbm
1&1\\
1&2\\
\ebm \bbm s_1\\s_2\ebm.$ The decoder then decodes $s_1$ and $s_2$. Again, this procedure is possible for any 2 contacted parties .

%


\vspace{0.2cm}
 
\noindent{\em 3) Perfect secrecy:}
At a high level, perfect secrecy is achieved here because each symbol in a share is ``padded" with at least one distinct key statistically independent of the secret,  making the shares of any party independent of the secret.
\end{example}

\subsection{Universal Staircase code construction}\label{sec:cons2}
We describe the $(n,k,z)$ Universal Staircase code construction that achieves optimal communication and read overheads $\CO(d)$ and $\RO(d)$ simultaneously for all possible values of $d$, i.e., $k+z\leq d\leq n$. Let $d_1=n,d_2=n-1,\dots,d_{h}=k+z$, with $h=n-k-z+1$, and  
 $\alpha_i=d_i-z,\ i=1,\dots, h$.   Choose $\alpha=LCM(\alpha_1,\alpha_2,\dots,\alpha_{h-1})$, that is the least common multiple of all the $\alpha_i$'s except for the last $\alpha_{h}=k$. The secret $\mathbf{s}$ consists of $k\alpha$ symbols $s_1,\dots,s_{k\alpha},$ uniformly distributed over $GF(q),\ q>n$, arranged in an $\alpha_1\times k\alpha/\alpha_1$ matrix $\cf$.

\noindent The construction uses $z\alpha$ iid random keys, drawn uniformly at random from $GF(q)$ and independently of the secret. The keys are partitioned into $h$ matrices $\ck_i, i=1,\dots,h,$ of respective dimensions $z\times k\alpha / \alpha_i\alpha_{i-1}$ (take $\alpha_0=1$). The matrices $\ck_1,\dots,\ck_i$ consist of  the overhead of keys decoded  by a user contacting $d_i$ parties.
We form $h$ matrices $M_i,$ $i=1,\dots,h,$ as follows,
\be
\normalsize
\linespread{1}
  \resizebox{0.93\textwidth}{!}{
\label{eq:matcon}
  \begin{tikzpicture}[baseline=(current  bounding  box.center)]
\tikzstyle{stealth} = [draw=none,text=black]
\node[stealth] (1) at (0,0){$
M_1=\hspace{0.3cm} \bbm
\multirow{2}{*}{$\cf$} \\
\\
 \ck_1\\
\ebm  \quad$,};

\normalsize
\linespread{1}

\node[stealth] (2) [right=0.3cm of 1] {$M_2=\hspace{0.3cm}\bbm
\cd_1 \\ \ck_2 \\ \mathbf{0} \ebm \quad,\ \dots \ ,$} ;
\node[stealth] (3) [right=0.3cm of 2] {$M_j=\hspace{0.3cm}\bbm
\cd_{j-1} \\ \ck_j \\ \mathbf{0} \ebm \quad, \  \dots \ ,$};
\node[stealth] (4) [right=0.3cm of 3] {$M_{h}=\hspace{0.3cm}\bbm
\cd_{h-1} \\ \ck_{h} \\ \mathbf{0} \ebm \quad.$};

\def\y{0.03}
\def\h{0.05}
\def\x{0.2}
\def\j{1.3}
\draw[<->,gray] (-0.08+0.15,-0.78-\y) to (0.82,-0.78-\y);
\node[stealth,gray] (4) at (0.45,-1-\y) {\scriptsize $k\alpha/\alpha_1$};

\draw[<->,gray,xshift=0.03cm] (3.2+\h+1*\x,-0.78-\y) to (3.93+\h+1*\x,-0.78-\y);
\node[stealth,gray] at (3.6+\h+1*\x,-1) {\scriptsize $ k\alpha/\alpha_1\alpha_2$};

\draw[<->,gray,xshift=0.11cm] (7.75+\h,-0.78-\y) to (3.65+3.88+\h+\j,-0.78-\y);
\node[stealth,gray,xshift=0.11cm] at (7+\h+\j,-1) {\scriptsize $ k\alpha/\alpha_{j-1}\alpha_{j}$};

\draw[<->,gray,xshift=0.25cm] (11.13+\h+\j,-0.78-\y) to (11.12+1.1+\h+\j,-0.78-\y);
\node[stealth,gray,xshift=0.25cm] at (11.72+\h+\j,-1) {\scriptsize $\alpha/\alpha_{h-1}$};

\draw[<->,gray] (-0.05,-0.7) to (-0.05,0.7);
\node[stealth,gray] at (-0.25,0) {\scriptsize $n$};
\draw[<->,gray] (0.85+\h,-.7) to (0.85+\h,0);
\draw[<->,gray] (0.85+\h,0) to (0.85+\h,0.7);
\node[stealth,gray] at (1+\h,-0.35) {\scriptsize $z$};
\node[stealth,gray] at (1.07+\h,0.35) {\scriptsize $\alpha_1$};

\draw[<->,gray,xshift=0.03cm] (3.15+0.2,-0.7) to (3.15+0.2,0.7);
\node[stealth,gray,xshift=0.05cm] at (3.15,0) {\scriptsize $n$};
\draw[<->,gray,xshift=0.01cm] (3.3+0.75+\h+1*\x,-.7) to (3.3+0.75+\h+1*\x,-0.25);
\draw[<->,gray,xshift=0.01cm] (3.3+0.75+\h+1*\x,-0.25) to (3.3+0.75+\h+1*\x,0.25);
\draw[<->,gray,xshift=0.01cm] (4.05+\h+1*\x,0.25) to (4.05+\h+1*\x,0.7);
\node[stealth,gray,xshift=0.01cm] at (3.35+0.85+\h+1*\x,0) {\scriptsize $z$};
\node[stealth,gray,xshift=0.01cm] at (3.35+0.92+\h+1*\x,0.45) {\scriptsize $\alpha_2$};
\node[stealth, gray,xshift=0.01cm] at (3.35+0.85+\h+1*\x,-0.45) {\scriptsize 1};

\draw[<->,gray,xshift=0.15cm] (6.15+0.2+\j,-0.7) to (6.15+0.2+\j,0.7);
\node[stealth,gray,xshift=0.15cm] at (6.15+\j,0) {\scriptsize $n$};
\draw[<->,gray,xshift=0.12cm] (7.45+\h+1*\x+\j,-.7) to (7.45+\h+1*\x+\j,-0.25);
\draw[<->,gray,xshift=0.12cm] (7.45+\h+1*\x+\j,-0.25) to (7.45+\h+1*\x+\j,0.25);
\draw[<->,gray,xshift=0.12cm] (7.45+\h+1*\x+\j,0.25) to (7.45+\h+1*\x+\j,0.7);
\node[stealth,gray,xshift=0.12cm] at (7.45+0.15+\h+1*\x+\j,0) {\scriptsize $z$};
\node[stealth,gray,xshift=0.12cm] at (7.45+0.21+\h+1*\x+\j,0.45) {\scriptsize $\alpha_j$};
\node[stealth, gray,xshift=0.12cm] at (7.45+0.45+\h+1*\x+\j,-0.45) {\tiny $n-d_j$};

\draw[<->,gray,xshift=0.25cm] (10.85+0.2+1*\j,-0.7) to (10.85+0.2+1*\j,0.7);
\node[stealth,gray,xshift=0.25cm] at (10.85+1*\j,0) {\scriptsize $n$};
\draw[<->,gray,xshift=0.22cm] (12.4+\h+1*\j,-.7) to (12.4+\h+1*\j,-0.25);
\draw[<->,gray,xshift=0.22cm] (12.4+\h+1*\j,-0.25) to (12.4+\h+1*\j,0.25);
\draw[<->,gray,xshift=0.22cm] (12.4+\h+1*\j,0.25) to (12.4+\h+1*\j,0.7);
\node[stealth,gray,xshift=0.22cm] at (12.4+0.15+\h+1*\j,0) {\scriptsize $z$};
\node[stealth,gray,xshift=0.22cm] at (12.4+0.15+\h+1*\j,0.45) {\scriptsize $k$};
\node[stealth, gray,xshift=0.22cm] at (12.4+0.25+0.15+\h+1*\j,-0.45) {\tiny $h-1$};

\end{tikzpicture}
}
\ee
Each matrix $\cd_j$ is formed of the $\left(n-j+1\right)^{th}$ row of $\bbm M_1 \ M_2 \dots M_j \ebm$ wrapped around to make a matrix of dimensions $\alpha_{j+1} \times k\alpha/\alpha_j \alpha_{j+1}$  for $j=1,\dots, h-1$. The $\mathbf{0}$'s are the all zero matrices used to complete the $M_{i}$'s to $n$ rows. The secret and the keys are arranged in the matrix $M=\bbm M_1\dots M_h\ebm$ defined in Table~\ref{tab:unim}.

\begin{table} \label{tab:unim}
\normalsize
\linespread{1}
\centering
\begin{tikzpicture}[baseline=(current  bounding  box.center)]
\tikzstyle{stealth} = [draw=none,text=black]

\def\x{-0.035}
\draw[draw=none,fill=xgray]  (-1.04+\x,-1.38) -- (-1.04+\x,-0.96) -- (-0.182+\x,-0.96) -- (-.182+\x,-.49) -- (.677+\x,-.49) -- 
(.677+\x,-.3) -- (1.491+\x,-0.3) -- (1.491+\x,0) -- (2.75,0) -- (2.75,-1.38) -- cycle; 

\node[stealth] at (0,0)
{$M=\left[\!
\begin{array}{c:c:c:c:c} 
 ~ & ~ &\multirow{2}{*}{$\cd_2$}  & \dots & \cd_{h-1}\\
 ~& \multirow{2}{*}{$\cd_1$} & & & \multirow{2}{*}{$\ck_h$} \\
\multirow{2}{*}{$\mathcal{S}$} & & \multirow{2}{*}{$\ck_3$}  &  \dots  & \\ 
 & \multirow{2}{*}{$\ck_2$} & & & \multirow{3}{*}{$\mathbf{0}$} \\
\multirow{2}{*}{$\ck_1$} &  & \multirow{2}{*}{$\mathbf{0}$} & \dots & \\
 & \mathbf{0} & & & \\
 \end{array}\!\right].$};
 
 \node[stealth] at (4.2,-0.5) {\Large $\substack{\text{\em staircase}\\ \text{\em structure}}$};
 \draw[->] (3.5,-0.6) ..controls (3,-1).. (2.4,-0.9);

\node[stealth,gray] at (-2.8,-0.4) {\scriptsize $n\times \alpha$};
\def\y{0.1}
\draw[thin] (-1.9+\x,-1.38) -- (-1.04+\x,-1.38) -- (-1.04+\x,-0.96) -- (-0.182+\x,-0.96) -- (-.182+\x,-.49) -- (.677+\x,-.49) -- 
(.677+\x,-.3) -- (1.491+\x,-0.3) -- (1.491+\x,0) -- (2.74,0); 
\draw[fill=white,draw=white] (-1.1,-1.386) -- (2.6,-1.386) -- (2.6,-1.49) -- (-1.1,-1.49) -- cycle;
\draw[thin, gray,decorate,decoration={brace,amplitude=5pt,mirror},yshift=-0.8pt]  (-1.93+\x,-1.5+\y) -- (-1.04+\x,-1.5+\y);
\node[stealth,color=gray] at (-1.45+\x,-2+\y) {$M_1$};
\draw[thin, gray,decorate,decoration={brace,amplitude=5pt,mirror},yshift=-0.8pt]   (-1.04+\x,-1.5+\y) -- (-0.182+\x,-1.5+\y) ;
\node[stealth,color=gray] at (-0.6,-2+\y) {$M_2$};
\draw[thin, gray,decorate,decoration={brace,amplitude=5pt,mirror},yshift=-0.8pt]  (-.182+\x,-1.5+\y) -- (.677+\x,-1.5+\y);
\node[stealth,color=gray] at (0.3+\x,-2+\y) {$M_3$};
\node[stealth,gray] at (1.1+\x,-2+\y) {$\dots$};
\draw[thin, gray,decorate,decoration={brace,amplitude=5pt,mirror},xshift=0.4pt,yshift=-0.8pt]  (1.446,-1.5+\y) -- (2.7,-1.5+\y);
\node[stealth,color=gray] at (2.12,-2+\y) {$M_h$};
\draw[thin] (-1.11,-1.38) -- (2.7,-1.38);

 \end{tikzpicture}
 \caption{ The structure of the matrix $M$ that contains the secret  and keys in the universal Staircase code construction. }
 \label{tab:unim}
\end{table}

The matrix $M$ is characterized by a special structure resulting from carefully choosing the entries of the $\cd_j$'s and placing the all zero sub-blocks in a staircase shape, giving these codes their name. This staircase shape allows to achieve optimal communication and read overheads $\CO$ and $\RO$ for all  possible $d$.

\noindent{\em Encoding: }The encoding is similar to the Staircase code construction. Let $V$ be an  $n\times n$ Vandermonde matrix defined over $GF(q)$. The matrix $M$, defined in Table~\ref{tab:unim}, is multiplied by $V$ to obtain the matrix  $C=VM$. The $n$ rows of $C$ form the $n$ different shares.

\vspace{0.2cm}

\noindent{\em Decoding:} To reconstruct the secret, a user contacting any $d_j$ parties indexed by $I\subseteq [n]$ downloads the first $k\alpha/\alpha_j$ symbols from each contacted party corresponding to $v_i \bbm M_1 \dots M_j \ebm$, for all $i\in I$.

We postpone the example of a Universal Staircase code to section~\ref{sec:uniex} to have it next to the proof of Theorem~\ref{thm:main2}.

\section{Staircase Code for fixed $d$}\label{sec:fds}
\begin{proof}[Proof of Theorem~\ref{thm:main}] \label{sec:proof1}
Consider the $(n,k,z,d)$ Staircase code defined in Section~\ref{sec:cons1}. We prove Theorem~\ref{thm:main} by establishing the following properties of the code:

\vspace{0.2cm}

\noindent{\em 1)} {\em Minimum $\CO(d)$ and $\RO(d)$:} We prove that a user contacting any $d$ parties can reconstruct the secret while incurring  minimum {\CO} and {\RO}. A user contacting any $d$ parties downloads the first $k$ symbols of each party. Let $I\subset[n],\ |I|=d,$ be the set of indices of the contacted parties, then the downloaded data is given by $V_{I}
\begin{bmatrix}
\mathcal{S}& \mathcal{R}_1
\end{bmatrix}^t,
$
where $V_I$ is a $d\times d$ square Vandermonde matrix formed of the rows of $V$ indexed by $I$, hence invertible. The user can always decode the secret (and the keys in $\ck_1$) by inverting $V_I$. The code is optimal on communication and read overheads $\CO(d)$ and $\RO(d)$, because the user only reads and downloads $kd$ symbols of size $kd/\alpha=kd/(d-z)$ units resulting in an overhead of $kd/\alpha-k=kz/\alpha=kz/(d-z)$ achieving the optimal $\CO(d)$ and $\RO(d)$ given in~\eqref{eq:CO}~and~\eqref{eq:RC}.

\vspace{0.2cm}

\noindent{\em 2)} {\em MDS property:} We prove that a user contacting any $t=k+z$ parties and downloading all their shares can reconstruct the secret.  Let $I\subset[n], |I|=t$, be the set of indices of the contacted parties. The information downloaded by the user is $V_{I}M$ and is given by,
\begin{equation*}
V_{I}\begin{bmatrix}
\mathcal{S} & \mathcal{D}\\
\multirow{2}{*}{$\mathcal{R}_1$}& \mathcal{R}_2\\
& \mathbf{0}
\end{bmatrix}.
\end{equation*}
First, we show that  the user can decode the entries of $\cd$ and $\ck_2$. The decoder considers the system,
\begin{equation}
\label{eq:all0}
V_{I}\begin{bmatrix}
\mathcal{D}&
\mathcal{R}_2&
\mathbf{0}
\end{bmatrix}^t=V'_{I}\begin{bmatrix}
\mathcal{D}&
\mathcal{R}_2
\end{bmatrix}^t.
\end{equation}
Recall that the dimensions of the all zero matrix in~\eqref{eq:all0} are $(\alpha-k)\times (\alpha-k)$, then $V'_{I}$ is a $(k+z)\times (k+z)$ square Vandermonde matrix formed by the first $(k+z)$ columns of $V_I$. Therefore, the user can always decode the entries of $\cd$ and $\ck_2$ because $V'_I$ is invertible.
Second, we prove that the user can always decode the entries of $\cf$ and hence reconstruct the secret. Recall that $\cal{D}$ is the transpose of the last $\alpha-k$ rows of $M_1\triangleq\bbm \mathcal{S} & \mathcal{R}_1\ebm^t$. By subtracting the previously decoded entries of $\cal{D}$ from 
$V_I\bbm
\mathcal{S}& \mathcal{R}_1
\ebm^t,$
the user obtains $V'_IM'_1$, where $V'_I$ is defined above and $M'_1$ is a $(k+z)\times k$ matrix formed by the first $k+z$ rows of $M_1$. Therefore, the user can always decode the entries of $M'_1$ because $V'_I$ is invertible. If $k+z\geq \alpha$, then $\cf$ is directly obtained since it is contained in $M'_1$. Otherwise, $M'_1$ consists of the first $k+z$ rows of $\cf$. The remaining rows of $\cf$ are contained in $\cd$ and were previously decoded. In both cases, the user can  decode all  the secret symbols $s_1,\dots,s_{k\alpha}$.
\vspace{0.2cm}

\noindent{\em 3)} {\em Perfect secrecy:}
We prove that for any subset $Z\subset [n]$, $\left\lvert Z \right\rvert =z$, the collection of shares indexed by $z$, denoted by $\ce_Z=\{w_i,i\in Z\}$, does not reveal any information about the secret as given in equation~\eqref{eq:secrecy}, i.e., $H(\mrS\mid \mrW_Z)=H(\mrS)$. Let $\mrR$ denote the random variable representing all the random keys, then it suffices to prove that $H(\mrR \mid \mrW_Z,\mrS)=0$ as detailed in the Appendix. Therefore, we need to  show that given the secret $\mathbf{s}$ as side information, any collection of $z$ shares can decode all the random keys. A collection of $\mathcal{W}_Z$ shares can be written as 
\begin{equation}\label{eq:vz}V_{Z}\begin{bmatrix}
\mathcal{S} & \mathcal{D}\\
\multirow{2}{*}{$\mathcal{R}_1$}& \mathcal{R}_2\\
& \mathbf{0}
\end{bmatrix},\end{equation}
where $V_Z$ is a $z\times d$ matrix corresponding to the rows of $V_Z$ indexed by $Z$. The linear system in~\eqref{eq:vz} can be divided into two systems as follows,
\begin{align}
V_{Z}& \begin{bmatrix} \mathcal{S}  & \mathcal{R}_1  \end{bmatrix}^t, \label{eq:adv_obs} \\
V_{Z} &\begin{bmatrix}  \mathcal{D} & \mathcal{R}_2 & \mathbf{0}\\ \end{bmatrix}^t.\label{eq:adv_ob2}
\end{align}
Given the secret as side information, it can be subtracted from  \eqref{eq:adv_obs}, which  becomes
\bes
V_{Z} \begin{bmatrix} \mathbf{0}  & \mathcal{R}_1 \end{bmatrix}^t=V''_Z \mathcal{R}_1,
\ees
where, $V''_Z$ is a $z\times z$ square Vandermonde matrix consisting of the last $z$ columns of $V_Z$. The entries of $\ck_1$ can always be decoded because $V''_Z$ is invertible. Now that $\ck_1$ is decoded and we have $\cf$ as side information, we can obtain $\cd$ as the last $\alpha-k$ rows of $\bbm \cf&\ck_1\ebm^t$. Then, the entries of $\cd$ are subtracted from the second system to obtain $V^*_Z\ck_2$, where $V^*_Z$ is a $z\times z $ square Vandermonde matrix consisting of the $(k+1)^{th}$ to the $(k+z)^{th}$ columns of $V_Z$. Hence, the entries of $\ck_2$ can always be decoded because $V^*_Z$ is invertible. Therefore, $H(\mrR \mid \mrW_Z,\mrS)=0$, $\forall \ Z,\ Z\subset [n],\ \left\lvert Z\right\rvert = z $ and perfect secrecy is achieved.

\end{proof}

\section{Universal staircase codes}\label{sec:uns}
\subsection{Example}
\label{sec:uniex}
We describe here the construction of an $(n,k,z)=(4,1,1)$ Universal Staircase code over $GF(q),\ q=5>n=4$, by following the construction in Section~\ref{sec:cons2}.
 We have $d_1=4,\ d_2=3$, $d_3=2$ and $\alpha_1=3$, $\alpha_2=2$, $\alpha_3=1$ and  $\alpha= LCM(\alpha_1,\alpha_2)=LCM(3,2)=6$. The secret $\mathbf{s}$ is formed of $k\alpha=6$ symbols uniformly distributed over $GF(5)$. The construction uses $z\alpha=6$ iid random keys drawn uniformly at random from $GF(5)$ and independently of the secret. The secret symbols and the random keys are arranged in the following matrices,
\bes
\cf=\bbm
s_1 & s_4\\
s_2 & s_5\\
s_3 & s_6\\
\ebm,
\quad \ck_1=\bbm
r_1 & r_2  \ebm,
\quad \ck_2=
\bbm r_3\ebm
\quad
\text{and } \quad  \ck_3=\bbm r_4&r_5&r_6\ebm.
\ees
To build the matrix $M$ which will be used for encoding the secret, we start with
\bes
M_1=\bbm
\multirow{3}{*}{$\mathcal{S}$}\\ \\ \\
\ck_1\\ 
\ebm
=\bbm
s_1 & s_4\\
s_2 & s_5\\
s_3 & s_6\\
r_1 & r_2\\
\ebm.
\ees
Then, $\cd_1$ is the $\alpha_2 \times k\alpha/\alpha_1\alpha_2=2\times1$ matrix that contains the symbols of the $n^{th}$ row of $M_1$, i.e., $\cd_1=\bbm
r_1& r_2\ebm^t$. Therefore, $M_2=\bbm \cd_1& \ck_2&\mathbf{0}\ebm^t=\bbm
r_1&r_2&r_3&0\ebm^t$. Similarly, we have $\cd_2=\bbm s_3&s_6&r_3\ebm$ and $M_3=\bbm s_3&s_6&r_3 \\ r_4&r_5&r_6\\ 0&0&0\\ 0&0&0\ebm$. We obtain $M$ by concatenating $M_1,\ M_2$ and $M_3$,
%
\begin{equation} \label{eq:unmat}
\linespread{1}
\begin{tikzpicture}[baseline=(current  bounding  box.center)]
\tikzstyle{stealth} = [draw=none,text=black]

\draw[draw=none, fill=xgray]  (-1.84,-0.92) -- (-0.345,-0.92) -- (-0.345,-0.475)  -- (0.4,-0.475) --(0.4,0) -- (2.52,0) -- (2.52,-0.92) -- cycle;

\node[stealth] (1) at (0,0) {$M=\left[\!
\begin{array}[h!]{cccccc}
s_1 & s_4 &\red r_1 & \blue s_3 &  \blue s_6 & \blue  r_3 \\
s_2 & s_5 & \red r_2 & r_4 &  r_5 &   r_6 \\
\blue s_3 &\blue  s_6 & \blue  r_3 &  0 &   0 &   0\\
\red r_1 &\red r_2 & 0 &   0 &   0 &   0 \\ 
\end{array}\!\right].$};

\def\y{0.05}
\draw[thin, gray,decorate,decoration={brace,amplitude=10pt,mirror},yshift=-0.4pt] (-1.76,-0.91-\y) -- (-0.345,-0.91-\y);
\node[stealth,color=gray] at (-1,-1.6) {$M_1$};
\draw[thin, gray,decorate,decoration={brace,amplitude=5pt,mirror},yshift=-0.4pt] (-0.345,-0.91-\y)  -- (0.4,-0.91-\y) ;
\node[stealth,gray] at (0.08,-1.6) {$M_2$};
\draw[thin, gray,decorate,decoration={brace,amplitude=10pt,mirror},yshift=-0.4pt] (0.4,-0.91-\y) -- (2.52,-0.91-\y);
\node[stealth,gray] at (1.5,-1.6) {$M_3$};
\draw[thin] (-1.76,-0.91) -- (-0.345,-0.91) -- (-0.345,-0.475)  -- (0.4,-0.475) --(0.4,0) -- (2.52,0) ;
\draw[thin](-0.345,-0.91)--(2.52,-0.91);
\end{tikzpicture}
\end{equation}
Here, $V$ is the $n\times n=4\times 4$ Vandermonde matrix over $GF(5)$ given in~\eqref{eq:unv}. The shares are given by the rows of the matrix $C=VM$ and shown in Table~\ref{tab:uniex}.

\begin{equation}
\label{eq:unv}
V= \begin{bmatrix}
1 & 1 & 1 & 1\\
1 & 2 & 4 & 3\\
1 & 3 & 4 & 2\\
1 & 4 & 1 & 4\\
\end{bmatrix}.
\end{equation}

\begin{table}[h!]
\centering
\linespread{1}
\normalsize
\begin{tabular}[h!]{c|c|c|c}
Party 1 & Party 2 & Party 3 & Party 4\\ \hline
$s_1+s_2+s_3+r_1$ & $s_1+2s_2+4s_3+3r_1$ &$s_1+3s_2+4s_3+2r_1$&$s_1+4s_2+s_3+4r_1$\\
$s_4+s_5+s_6+r_2$ & $s_4+2s_5+4s_6+3r_2$ & $s_4+3s_5+4s_6+2r_2$ & $s_4+4s_5+s_6+4r_2$ \\
\red $r_1+r_2+r_3$ &\red  $r_1+2r_2+4r_3$ & \red $r_1+3r_2+4r_3$ & \red $r_1+4r_2+r_3$ \\
\blue $s_3+r_4$ & \blue $s_3+2r_4$ & \blue $s_3+3r_4$ & \blue $s_3+4r_4$ \\
\blue $s_6+r_5$ & \blue $s_6+2r_5$ & \blue $s_6+3r_5$ & \blue $s_6+4r_5$ \\
\blue $r_3+r_6$ & \blue $r_3+2r_6$ & \blue $r_3+3r_6$ & \blue $r_3+4r_6$ \\
\end{tabular}
\caption{An example of a universal Staircase code for $(n,k,z)=(4,1,1)$ over $GF(5)$.}
\label{tab:uniex}
\end{table}
\noindent The constructed Universal Staircase code satisfies the following properties:
\vspace{0.2cm}

\noindent{\em 1) MDS:} We check that a user contacting $d_3=k+z=2$ parties can decode the secret. Suppose that  the user contacts parties $1$ and $2$. The data downloaded by the user is $V_{\{1,2\}}M$ and is given by,

\begin{equation}
\label{eq:mdsuni}
\begin{tikzpicture}[baseline=(current  bounding  box.center)]
\tikzstyle{stealth} = [draw=none,text=black]
\linespread{1}
\node[stealth] (1) at (0,0) {$\begin{bmatrix}
1 & 1 & 1 & 1 \\
1 & 2 & 4 & 3 \\
\end{bmatrix}$};
\node[stealth] (2) [right=-0.1cm of 1] {$
\left[\!
\begin{array}[h!]{cc:c:ccc}
s_1 & s_4 & r_1 & s_3 & s_6 & r_3 \\
s_2 & s_5 & r_2 & r_4 & r_5 & r_6 \\
s_3 & s_6 & r_3 &  0 &   0 &   0\\
r_1 & r_2 & 0 &   0 &   0 &   0 \\ 
\end{array}\!\right].$};

\def\y{0.05}
\draw[thin, gray,decorate,decoration={brace,amplitude=10pt,mirror},xshift=3.2cm,yshift=-0.4pt] (-1.76,-0.91-\y) -- (-0.32,-0.91-\y);
\node[stealth,color=gray,xshift=3.2cm,] at (-1,-1.6) {$M_1$};
\draw[thin, gray,decorate,decoration={brace,amplitude=5pt,mirror},xshift=3.2cm,,yshift=-0.4pt] (-0.32,-0.91-\y)  -- (0.4,-0.91-\y) ;
\node[stealth,gray,xshift=3.2cm,] at (0.08,-1.6) {$M_2$};
\draw[thin, gray,decorate,decoration={brace,amplitude=10pt,mirror},xshift=3.2cm,,yshift=-0.4pt] (0.4,-0.91-\y) -- (2.52,-0.91-\y);
\node[stealth,gray,xshift=3.2cm,] at (1.5,-1.6) {$M_3$};

\draw[thin, gray,decorate,decoration={brace,amplitude=10pt,mirror},xshift=0.4pt,yshift=-0.4pt] (-1,-0.5) -- (1,-0.5);
\node[stealth,gray] at (0,-1.2) {$V_{\{1,2\}}$};

%
%


\end{tikzpicture}
\end{equation}
We will show that the user can decode the secret by successively solving the linear systems $V_{\{1,2\}}M_3$, $V_{\{1,2\}}M_2$ and $V_{\{1,2\}}M_1$. The decoder starts by considering  $V_{\{1,2\}}M_3$ which gives,
\begin{equation}
\label{eq:dec}
\begin{bmatrix}
1 & 1 \\
1 & 2 \\
\end{bmatrix}
\begin{bmatrix}
s_3 & s_6 & r_ 3 \\
r_4 & r_5 & r_6 \\
\end{bmatrix}.
\ee
The matrix on the left is invertible, and the user can decode the secret symbols and keys  in \eqref{eq:dec}.   
Then, the decoder considers the system $V_{\{1,2\}}M_2$ after subtracting from it the value of $r_3$ decoded in the previous step. The obtained system is again invertible and the decoder can decode $r_1$ and $r_2$. The decoder then considers $V_{\{1,2\}}M_1$, after canceling out $r_1,\ r_2,\ s_3,\ s_6$ decoded so far, to obtain the following system,
\bes
\begin{bmatrix}
1 & 1 \\
1 & 2 \\
\end{bmatrix}
\begin{bmatrix}
s_1 & s_4\\
s_2 & s_5 \\
\end{bmatrix}.
\ees
The matrix on the left is again invertible and the decoder can reconstruct the secret. This remains true irrespective of which 2 parties are contacted.
\vspace{.2cm}

\noindent{\em 2) Minimum {\CO} and {\RO} for $d_2=3$ and $d_1=4$:}
We check that a user contacting any $d,\ d=3,4,$ parties can decode the secret while achieving the minimum communication and read overheads given in~\eqref{eq:CO}~and~\eqref{eq:RC}. Suppose a user contacts $d_2=3$ parties indexed by $I\subset [n]$. The user  reads and downloads the first $k\alpha/\alpha_2=3$ symbols of each contacted share corresponding to $V_{I} \bbm M_1\ M_2 \ebm$ (in black and red), where $V_I$ is the matrix formed by the rows of $V$ indexed by $I$. The user will be able to reconstruct the secret by implementing a decoding procedure similar to the one above. The resulting {\CO} and {\RO} are equal to $3/2-k=1/2$ units achieving the optimal $\CO(d_2)$ and $\RO(d_2)$ given in~\eqref{eq:CO}~and~\eqref{eq:RC}. 
In the case when  a user contacts $d_1=4$ parties, the user reads and downloads the first $k\alpha/\alpha_1=2$ symbols of each contacted share corresponding to $V_IM_1$ (in black). The user can always decode the secret because $V_I$ here is a $4\times 4$ square Vandermonde matrix, hence invertible. The resulting {\CO} and {\RO} are equal to $1/3$ achieving the optimal $\CO(d_1)$ and $\RO(d_1)$ given in~\eqref{eq:CO}~and~\eqref{eq:RC}.

\vspace{0.2cm}

\noindent{\em 3) Perfect secrecy:}  At a high level, perfect secrecy is achieved here because each symbol in a share is ``padded" with at least one distinct key statistically independent of the secret,  making the shares of any party independent of the secret.

\subsection{Proof of Theorem~\ref{thm:main2}} \label{sec:proof2}

Consider the $(n,k,z)$ Universal Staircase code construction defined in Section~\ref{sec:cons2}. We prove Theorem~\ref{thm:main2} by establishing the following properties.

\vspace{0.2cm}

\noindent{\em 1)} {\em Encoding is well defined:} We prove that the  $(n-j+1)^{th}$ row of $\bbm M_1 \dots M_j \ebm$ has the same number of entries as $\cd_j, j=1,\dots,h-1$. Therefore, we  can always construct the matrix $\cd_j$.  In fact, the number of entries of one row of $\bbm M_1\dots M_{j}\ebm$  is equal to the sum of the number of columns of the $M_i$'s, $i=1,\dots,j$. Notice that $\alpha_{i-1}=\alpha_{i}+1$, then we can write, \bes
\dfrac{k\alpha}{\alpha_i\alpha_{i-1}}=k\alpha\left(\dfrac{1}{\alpha_{i}}-\dfrac{1}{\alpha_{i-1}}\right).
\ees
Hence, the number of columns of $\bbm M_1 \dots M_j\ebm$ is given by,
\begin{align} \label{eq:rows}
\dfrac{k\alpha}{\alpha_1}+k\alpha\left(\dfrac{1}{\alpha_2}-\dfrac{1}{\alpha_1}\right)+\dots+k\alpha\left(\dfrac{1}{\alpha_j}-\dfrac{1}{\alpha_{j-1}}\right)&=\dfrac{k\alpha}{\alpha_j},
\end{align}
 which is equal to the number of entries of $\cd_j$.

\vspace{0.4cm}

\noindent{\em 2)} {\em MDS and minimum $\CO(d)$ and $\RO(d)$ for all $d$, $k+z\leq d\leq n$:}
We prove that a user contacting any $d,\ k+z\leq d \leq n,$ parties can decode the secret while achieving the minimum communication and read overheads given in~\eqref{eq:CO}~and~\eqref{eq:RC}. Notice that the MDS property follows directly from the fact that a user contacting $d_h=k+z$ parties can reconstruct the secret by reading and downloading all the contacted shares. 

A user contacting any  $d_j$, $j=1,\dots, h,$ parties downloads the first $k\alpha/\alpha_j$ symbols of each party. Let $I\subseteq [n],\ |I|=d_j,$ be the set of indices of the contacted parties and let $V_I$ be the matrix formed of the rows of $V$ indexed by $I$. The total downloaded data is given by $V_{I}\bbm M_1 \dots M_j\ebm$ and can be divided into $j$ linear systems given as follows,
\begin{align}
V_IM_1&=V_{I}\bbm \cf & \ck_1 \ebm^t \label{eq:undec1}\\
V_IM_2&=V_{I}\bbm \cd_{1}& \ck_2 & \mathbf{0}\ebm^t \\
~&\vdots  \nonumber\\ 
V_IM_{j-1}&=V_{I}\bbm \cd_{j-2}&\ck_{j-1} & \mathbf{0}\ebm^t \label{eq:undec2}\\
V_IM_j&=V_{I}\bbm \cd_{j-1}& \ck_j & \mathbf{0}\ebm^t. \label{eq:undec3}
\end{align}

We  prove by induction that the user can always reconstruct the secret by iteratively decoding $M_i$, $i=j,\dots,1,$ in each linear system $V_IM_i$. To that end, we verify the induction hypothesis for $i=j$. Given the system in~\eqref{eq:undec3}, we show that the user can always decode $M_j$. The zero block matrix in~\eqref{eq:undec3} is of dimensions $(n-d_j)\times (k\alpha/\alpha_j\alpha_{j-1})$. Therefore, \eqref{eq:undec3} can be rewritten as $V'_{I}\bbm \cd_{j-1}& \ck_j \ebm$, where $V'_{I}$ is the square Vandermonde matrix of dimensions $d_j\times d_j$  formed by the first $d_j$ columns of $V_{I}$. Hence, the user can always decode the entries of $M_j$ by inverting $V'_I$.

Next, suppose that the user can decode all the $M_i$'s, $i=j,\dots,l+1$, we prove that the user can always decode $M_l$. The $l^{th}$ system is given by $V_IM_l$. By construction $M_l$ contains $d_l$ non-zero rows, because the $\mathbf{0}$ block matrix is of dimensions $(n-d_l)\times (k\alpha/\alpha_l\alpha_{l-1})$. In addition, the entries of the last $l-1$ non-zero rows of $M_j$ are present in $\cd_f$ for $f=j-1,\dots,l-1,$ which were previously decoded. It can be checked that $d_j=d_l-(l-1)$ for all $l<j$.
Therefore, after subtracting the last $l-1$ rows of $M_l$, the system becomes $V'_IM'_l$, where $V'_I$ is again the $d_j\times d_j$ square Vandermonde matrix defined above and $M'_l$ is the matrix formed of the first $d_j=d_l-(l-1)$ rows of $M_l$. Henceforth, the user can always decode $M'_l$ by inverting $V'_I$. Finally, the user can decode all the entires of $M_l$ that consist of the entries of $M'_l$ and the entries of the last $l-1$ rows of $M_l$, which were previously decoded.

Next, we show that minimum {\CO} and {\RO} are achieved. The number of symbols read and downloaded by a user contacting $d_j$ parties is equal to $d_j(k\alpha/\alpha_j)$ symbols which corresponds to $d_j k/\alpha_j$ units. Then, the communication and read overheads are given by  $d_j k/\alpha_j-k=kz/\alpha_j=kz/(d_j-z)$, which matches the optimal $\CO(d_j)$ and $\RO(d_j)$ for all $d_j=k+z,\dots,n,$ given in~\eqref{eq:CO}~and~\eqref{eq:RC}.

\vspace{0.2cm}

\noindent{\em 3)} {\em Perfect secrecy:} Similarly to the proof of perfect secrecy in Theorem~\ref{thm:main}, we need to show that $H(\mrR \mid \mrW_Z,\mrS)=0$ for all $Z\subset [n]$, $|Z|= z$ (see Appendix). This is equivalent to showing that given the secret $\mathbf{s}$ as side information, any collection $\mathcal{W}_Z$ of $z$ shares can decode all the random keys. A collection of $\mathcal{W}_Z$ of $z$ shares can be written as $V_Z\bbm M_1 \dots M_h\ebm$, which can be divided into $h=n-k-z+1$ linear systems as follows,
\begin{align}
V_ZM_1&=V_{Z} \begin{bmatrix} \mathcal{S}  & \mathcal{R}_1  \end{bmatrix}^t\label{eq:unadv_obs1}\\
V_ZM_2&=V_{Z} \begin{bmatrix}  \mathcal{D}_1 & \mathcal{R}_2 & \mathbf{0}\\ \end{bmatrix}^t\label{eq:unadv_ob1}\\
~&\vdots \nonumber\\
V_ZM_h&=V_{Z} \begin{bmatrix}  \mathcal{D}_{h-1} & \mathcal{R}_{h} & \mathbf{0}\\ \end{bmatrix}^t.\label{eq:unadv_ob2}
\end{align}

We will prove by induction that given the secret $\mathbf{s}$ as side information, any collection $\ce_Z$ of $z$ shares can always iteratively decode $\ck_i$, $i=1,\dots,h,$ in each linear system $V_ZM_i$. To that end, we verify the induction hypothesis for $i=1$ by showing that a collection of $\ce_Z$ shares can always decode $\ck_1$  in~\eqref{eq:unadv_obs1}. Recall that the dimensions of $\ck_1$ are $z\times k\alpha/\alpha_1$. Given the secret $\mathbf{s}$, \eqref{eq:unadv_obs1} becomes,
\bes
V_{Z} \begin{bmatrix} \mathbf{0}  & \mathcal{R}_1  \end{bmatrix}^t=V''_Z \mathcal{R}_1,
\ees
where $V''_Z$ is a $z\times z$ square Vandermonde matrix formed by the last $z$ columns of $V_Z$. Therefore, $\ck_1$ can be decoded  by inverting $V''_Z$. 

Next, we suppose that any collection of $\ce_Z$ shares can decode all the  $\ck_i$'s  for $i=1,\dots,l-1$, and show that any collection of $\ce_Z$ can decode  $\ck_l$. The $l^{th}$ system is given by $V_IM_l=V_I\bbm \cd_{l-1}& \ck_l & \mathbf{0}\ebm^t$. By construction, $\cd_{l-1}$ consists of the entries of the last row of $M_{l-1}$ which were previously decoded. Given the previously decoded information, any collection of $\ce_Z$ shares can cancel out the entries of $\cd_{l-1}$ to obtain $V^*_Z\ck_{l}$. Since the dimensions of $\ck_l$ are $z\times k\alpha/\alpha_{l}\alpha_{l-1}$, the matrix $V^*_Z$ is a $z\times z$ square Vandermonde matrix formed by the $(\alpha_l+1)^{th}$ to $(\alpha_l+z)^{th}$ rows of $V_Z$. Thus, $\ck_l$ can be always  decoded because $V^*_Z$ is invertible. Therefore, all the keys can always be decoded. Hence, $H(\mrR \mid \mrW_Z,\mrS)=0$. This concludes the proof of Theorem~\ref{thm:main2}.

\vspace{0.2cm}

\noindent {\em $\Delta$-Universal Staircase codes: }Note that the construction of Universal Staircase codes can be modified to construct Staircase codes that achieve minimum {\CO} and {\RO} only for a desired  subset $\Delta$ of all possible $d$'s, i.e., $\Delta \subseteq\{k+z,\dots,n\}$.  We refer to these codes as $(n,k,z,\Delta)$ $\Delta$-universal Staircase codes. The advantage of these codes over  universal codes is that they may require smaller  number of symbols per share $\alpha$.

 \noindent{\em Encoding:} Let  $\Delta'\triangleq \Delta\setminus \{k+z\}$ and  order the $d$'s in $\Delta'$ in decreasing order. We write $\Delta'=\{d_{i_1},\dots,d_{i_{|\Delta'|}}\}\subseteq \{d_1,\dots,d_{h-1}\}$, where $d_{i_1}>d_{i_2}>\dots>d_{i_{|\Delta'|}}$. Let $\alpha_{i_j}=d_{i_j}-z$ for all $d_{i_j}\in \Delta'$ and let $\alpha=LCM(\alpha_1,\dots,\alpha_{\left\lvert \Delta' \right\rvert})$.
Define $d_{i_{|\Delta'|+1}}\triangleq k+z$ and $\alpha_{i_{|\Delta'|+1}}\triangleq k$. The secret symbols are arranged in a matrix $\cf$ of dimensions $\alpha_{d_{i_1}}\times k\alpha/\alpha_{d_{i_1}}$ and the random keys are partitioned into the matrices $\ck_{i_1},\dots,\ck_{i_{|\Delta'|+1}},$ of dimensions $z\times k\alpha/\alpha_{i_1}$ for $\ck_{i_1}$ and $z\times k\alpha(\alpha_{i_j}-\alpha_{i_{j-1}})/(\alpha_{i_j}\alpha_{i_{j-1}})$ for all other $\ck_{i_j}$, $j=2,\dots,|\Delta'|+1$.
Construct $M_{i_1}$ as the $d_{i_1}\times k\alpha/\alpha_{i_1}$ matrix structured as $M_1$ in~\eqref{eq:matcon}. And, for each $d_{i_j}$, $j=2,\dots, |\Delta'|+1$, construct $M_{i_j}$ as the $d_{i_1}\times k\alpha(\alpha_{i_j}-\alpha_{i_{j-1}})/(\alpha_{i_j}\alpha_{i_j-1})$ structured as $M_{i_j}$ in~\eqref{eq:matcon}. The matrix $\cd_{i_j}$, $j=1,\dots,|\Delta'|$, is the matrix of dimensions $\alpha_{i_{j+1}}\times k\alpha(\alpha_{i_{j+1}}-\alpha_{i_{j}})/(\alpha_{i_{j+1}}\alpha_{i_{j}})$ containing the last $d_{i_j}-d_{i_{j+1}}$ rows of $\bbm M_{i_1} \dots M_{i_{j}}\ebm$, from row $d_{i_j}$ to row $d_{i_{j+1}}+1$. Then, concatenate the constructed matrices, $M_{i_1},\dots,M_{i_{|\Delta'|+1}},$ to obtain the matrix $M$ of dimensions  $d_{i_1}\times \alpha$. The matrix $M$ is multiplied by a Vandermonde matrix of dimensions $n\times d_{i_1}$ to obtain the shares.
\vspace{.2cm}

\noindent{\em Decoding:} To reconstruct the secret, a user contacting any $d_{i_j}$ parties,  indexed by $I\subseteq[n]$,   downloads the first $k\alpha/\alpha_{i_j}$ symbols from each contacted party corresponding to $v_i \bbm M_{i_1} \dots M_{i_j} \ebm$ , for all $i\in I$.

\begin{corollary} \label{cor:1}
 Let $\Delta\subseteq \{k+z,\dots,n\}$. The $(n,k,z,\Delta)$ $\Delta$-universal Staircase codes defined above over $GF(q)$, $q>n$, satisfies the required MDS and perfect secrecy constraints given in~\eqref{eq:secrecy}~and~\eqref{eq:mds} and achieves optimal communication overhead $\CO(d)$ and read overhead $\RO(d)$ given in~\eqref{eq:CO}~and~\eqref{eq:RC} simultaneously for all $d$, $d\in \Delta$.
\end{corollary}
\noindent We omit the proof of Corollary~\ref{cor:1} since it follows the same steps of the proof of Theorem~\ref{thm:main2}. 
{
\section{Threshold changeable secret sharing} \label{sec:tcss}
An $(n,k,z;t')$ threshold changeable secret sharing (TCSS) code,  defined in \cite{MPSH99}, is an $(n,k,z)$ secret sharing scheme (satisfying \eqref{eq:secrecy}~and \eqref{eq:mds}), where the threshold $t=k+z$ can be changed to $t'>t$ in a decentralized way without the dealer. The parties are allowed to communicate as long as the security constraint is not violated.
 The efficiency of a TCSS is measured by the new share size for the new threshold  $t'$, which we refer to as the storage cost ({\SC}) of the scheme\footnote{Any secret sharing scheme is trivially threshold changeable, because a user contacting $t'>t$ parties can decode the secret by downloading any $t$ shares. However, it does not achieve minimum storage cost for the new threshold.}. Different variants of threshold changeable secret sharing schemes have been studied in the literature, see e.g., \cite{TCSS1,TCSS2,TCSS3}. A connection between TCSS and CESS is shown in \cite{WW08}. Code constructions are provided in \cite{MPSH99,WW08,ZYSMH12} for the case when $z=t-1$ and the threshold $t'$ is given a priori.

In this section, we show how to construct an $(n,k,z;t')$ TCSS code for a given $t'>t$ using an $(n,k,z,d=t')$ Staircase code. However, different values of $t'$ for the same $(n,k,z)$ may require different Staircase codes. We show that this can be avoided by constructing what we call an $(n,k,z;[t+1:n])$ Universal TCSS code using an $(n,k,z)$ Universal Staircase code. Both constructions involve the parties deleting parts of their shares and do not require communication among the parties. Moreover, this construction achieves  the optimal storage cost ({\SC}) 
\begin{equation}\label{eq:tc1}
\SC = \dfrac{k}{t'-z},
\end{equation}
which is the minimum share size required if the dealer were present. The next example shows how to construct an $(n,k,z;[t+1:n])$ Universal TCSS code with optimal {\SC} from an $(n,k,z)$ Universal Staircase code by deleting  parts of each share.

\begin{example}\label{ex:tcss}
Consider the problem of constructing an $(n,k,z;[t+1:n])=(4,1,1;[3:4])$ Universal TCSS code for all possible $t'$, i.e., $t'=3$ and $4$. To this end, we use an $(n,k,z)=(4,1,1)$ Universal Staircase code constructed in Section~\ref{sec:uniex}. The shares given to each party are depicted in Table~\ref{tab:tcss}.

{\color{black}
\begin{table}[h!]
\normalsize
\linespread{1}
\begin{tikzpicture}
\definecolor{gainsboro}{rgb}{0.7, 0.7, 0.7}
\def \nco{\color{blue}}
\def \white{\color{white}}


\node at (0,0) {
\begin{tabular}[h!]{c|c|c|c}
Party 1 & Party 2 & Party 3 & Party 4\\ \hline
$ s_1+s_2+s_3+r_1$ & $s_1+2s_2+4s_3+3r_1$ &$s_1+3s_2+4s_3+2r_1$&$s_1+4s_2+s_3+4r_1$\\
$s_4+s_5+s_6+r_2$ & $s_4+2s_5+4s_6+3r_2$ & $s_4+3s_5+4s_6+2r_2$ & $s_4+4s_5+s_6+4r_2$ \\
\red $r_1+r_2+r_3$ &\red  $r_1+2r_2+4r_3$ & \red $r_1+3r_2+4r_3$ & \red $r_1+4r_2+r_3$ \\
\nco $s_3+r_4$ & \nco $s_3+2r_4$ & \nco $s_3+3r_4$ & \nco $s_3+4r_4$ \\
\nco $s_6+r_5$ & \nco $s_6+2r_5$ & \nco $s_6+3r_5$ & \nco $s_6+4r_5$ \\
\nco $r_3+r_6$ & \nco $r_3+2r_6$ & \nco $r_3+3r_6$ & \nco $r_3+4r_6$ \\
\end{tabular}
};

\draw[thin, gray,decorate,decoration={brace,amplitude=5pt,mirror}] (-7.2,1.1) -- (-7.2,-0.2);

\draw[dashed,fill=white,opacity=0.6] (-7.2,-1.8) rectangle (7,-0.28); 
\node[font=\footnotesize,opacity=0.4] at (-7.8,0.45) {\rotatebox{90}{\begin{varwidth}{1.5cm}\em New share for $t'=3$\end{varwidth}} };
\node[font=\footnotesize,opacity=0.4] at (-7.6,-1) {\rotatebox{90}{\begin{varwidth}{1.5cm}\em Deleted\end{varwidth}} };
\end{tikzpicture}
\caption{A $(4,1,1;[3:4])$ Universal TCSS code  obtained from an $(4,1,1)$ Universal Staircase code over $GF(5)$. The original code has threshold $t=k+z=2$ and can be changed to either $t'=3$ or $4$.  The threshold change from $t=2$ to $t'=3$ is depicted. Each party deletes the last $3$ symbols of its share. Similarly, the threshold can be changed to $t'=4$ by keeping the first two symbols of each share. In both cases, the obtained code  achieves minimum storage cost ({\SC}) (share size) given by \eqref{eq:tc1}.}
\label{tab:tcss}
\end{table}
}
In our construction, to change the threshold from $t=k+z$ to any $t'$, $t' \in \{t+1,\dots,n\}$, each party deletes the last $\dfrac{t'-z-k}{t'-z}\alpha$ symbols of its share. Recall that in CESS, each share is of unit size and consists of $\alpha$ symbols ($\alpha$ symbols = $1$ unit). In this example, to change the threshold from $t=2$ to $t'=3$, each party deletes the last $3$ symbols  (in shaded blue) of its share. The obtained code achieves the minimum Storage Cost ({\SC}) given in~\eqref{eq:tc1}, because each new share is of size $3$ symbols equal to $1/2$ unit. One can verify that a user contacting any $t'=3$ parties and downloading their new shares can decode the secret.

Similarly, the same code can be used to change the threshold from $t=2$ to  $t'=4$. Each party deletes the last $4$ symbols (in red and shaded blue) of its original share (or deletes the last symbol, in red, if the threshold was already changed to 3). Each new share consists of $2$ symbols equal to $1/3$ unit. Hence, the obtained code achieves minimum Storage Cost ({\SC}) given in~\eqref{eq:tc1}. One can verify that a user downloading all the shares can decode the secret. In both cases, secrecy is inherited from the Staircase code, because the parties do not exchange any information when changing the threshold.
\end{example}


\begin{corollary} \label{thm:tcmain}
An $(n,k,z;t')$ TCSS code, respectively an $(n,k,z;[t+1:n])$ Universal TCSS code, can be constructed using an $(n,k,z,d)$ Staircase code defined in Section~\ref{sec:cons1}, respectively an $(n,k,z)$ Universal Staircase code defined in Section~\ref{sec:cons2}. To change the threshold from $t=k+z$ to $t'$, each party deletes the last $\dfrac{t'-k-z}{t'-z}\alpha$ symbols of its share. Both constructions achieve optimal storage cost ({\SC}) given in \eqref{eq:tc1}.
 \end{corollary}

\begin{proof}
We prove that an $(n,k,z;[t+1:n])$ Universal TCSS code can be constructed using an $(n,k,z)$ Universal Staircase code and omit the proof for $(n,k,z;t')$ TCSS code, since it follows the same steps.

\vspace{0.2cm}

Starting with an $(n,k,z)$ Universal Staircase code, the threshold is $t=k+z$. Assume that the threshold is to be changed to $t'$ for any $t' \in \{t+1,\dots,n\}$. Each party deletes the last $\dfrac{t'-z-k}{t'-z}\alpha$ symbols of its share (original share size is $\alpha$ symbols). 

We establish the following properties.
\begin{enumerate}
\item {\em Minimum Storage Cost ({\SC}):} By construction, the new share size is $\alpha-(t'-z-k)\alpha/(t'-z)=k\alpha/(t'-z)$ symbols. Recall that each $\alpha$ symbols are equal to 1 unit, hence each share is of size $k/(t'-z)$ units and~\eqref{eq:tc1} is achieved.

\item {\em MDS in $t'$:} By construction, after changing the threshold to $t'$, each party keeps exactly the symbols that are sent to a user contacting any $t'$ parties in the original CESS code. Therefore, the user can decode the secret.

\item {\em Perfect secrecy:} Since the parties do not exchange any information when changing the threshold, perfect secrecy follows from the properties of the original  Universal Staircase code.
\end{enumerate}
\end{proof}

\begin{remark}
Note that the Universal TCSS code obtained from our construction also minimizes the communication and read overheads ({\CO} and {\RO}) in addition to minimizing the storage cost ({\SC}). In other words, the new $n$ shares stored after the threshold update, allow a user contacting any $d$ parties, $d\in \{t',\dots,n\},$ to decode the secret while achieving the minimum communication and read overheads given in~\eqref{eq:CO}~and~\eqref{eq:RC}.

For instance, in Example~\ref{ex:tcss} for the new threshold $t'=3$, a user contacting any $d=4$ parties and downloading the first two symbols (in black) of each new share can decode the secret. The incurred {\CO} (and {\RO}) is equal to $2$ symbols equal to $1/3$ unit and is minimum, i.e., achieves~\eqref{eq:CO}~and~\eqref{eq:RC}.
\end{remark}
}

\section{Conclusion}\label{sec:conc}

We considered the communication efficient secret sharing (CESS) problem. The goal is to minimize the read and download overheads for a user interested in decoding the secret. To that end, 
we introduced a new class of deterministic linear CESS codes,  called  {\em Staircase Codes}. We described two explicit constructions of Staircase codes. The first construction achieves minimum overhead for any given number of parties $d$ contacted by the user. The second is a universal construction that achieves minimum overheads simultaneously for all possible values of $d$.
The introduced codes require a small finite field $GF(q)$ of size $q> n$, which is the same requirement for Reed Solomon based SS codes \cite{McESa81}.  Finally, we described how Staircase codes can be used to construct threshold changeable secret sharing (TCSS) codes.

In conclusion, we  point out some problems that remain open. The model we considered here and the proposed Staircase codes   can provide security against parties corrupted by a passive Eavesdropper. However, the problem of constructing communication and read efficient codes that provide security against an active (malicious) adversary remains open.  Moreover,  constructing  threshold changeable secret sharing codes where the security level can be increased by increasing the number of possibly colluding parties   also remains open in general (only special cases   were solved in  \cite{TCSS2}).
%

\appendix
Let $W_i$ denote the random variable representing share $w_i$, and for any subset $B\subseteq \{1,\dots,n\}$ denote by $\mrW_B$ the set of shares indexed by $B$, i.e., $\ce_B=\{W_i; i\in B\}$. We prove that, for all $Z\subset\{1,\dots,n\},\ |Z|= z$, the perfect secrecy constraint $H(\mrS \mid \mrW_Z)=H(\mrS)$, given in~\eqref{eq:secrecy},  is equivalent to $H(\mrR \mid \mrW_Z,\mrS)=0$. The proof is standard \cite{RSKGlobecom11,RSS12} but we reproduce it here for completeness.  In what follows, the logarithms in the entropy function are taken   base $q$. We can write,
\begin{align}
H(\mrS \mid \mrW_Z)&=H(\mrS)-H(\mrW_Z)+H(\mrW_Z \mid \mrS)\\
~&=H(\mrS)-H(\mrW_Z)+H(\mrW_Z \mid \mrS)-H(\mrW_Z \mid \mrS, \mrR) \label{eq:sk}\\
~&=H(\mrS)-H(\mrW_Z)+I(\mrW_Z; \mrR \mid \mrS)\\
~&=H(\mrS)-H(\mrW_Z)+H(\mrR \mid \mrS) - H(\mrR \mid \mrW_Z,\mrS)\\
~&=H(\mrS)-H(\mrW_Z)+H(\mrR \mid \mrS) \label{eq:zs}\\
~&=H(\mrS)-H(\mrW_Z)+H(\mrR)\\
~&=H(\mrS)-z\alpha+z\alpha \label{eq:sc}\\
~&=H(\mrS).
\end{align}
\vspace{0.2cm}

\noindent Equation~\eqref{eq:sk} follows from the fact that given the secret $\mathbf{s}$ and the keys $\ck$ any share can be decoded, equation~\eqref{eq:zs} follows because $H(\mrR \mid \mrS, \mrW_Z)=0$ and equation~\eqref{eq:sc}  follows because the Staircase code constructions use $z\alpha$ independent random keys.

\bibliographystyle{ieeetr}
\bibliography{IEEEabrv,DSS}

%

\end{document}